\documentclass[runningheads]{llncs}
\usepackage[T1]{fontenc}
\usepackage{amsfonts}
\usepackage{amsmath}
\usepackage{amssymb}
\usepackage{csquotes}
\usepackage[colorlinks=true, allcolors=blue]{hyperref}
\begin{document}
\title{Cryptanalysis of the Legendre Pseudorandom Function over Extension Fields}
\author{Daksh Pandey}
\authorrunning{Daksh P.}
\institute{Indian Institute of Technology, Roorkee}
\maketitle              
\begin{abstract}
The Legendre Pseudorandom Function (PRF) is a highly efficient cryptographic primitive built upon the Legendre symbol, recently gaining significant traction due to its low multiplicative complexity in Multi-Party Computation (MPC) and Zero-Knowledge Proof (ZKP) protocols. While the security of the Legendre PRF over prime fields $\mathbb{F}_p$ has been extensively studied, recent interest has shifted toward its instantiation over extension fields $\mathbb{F}_{p^r}$. In this paper, we present the first comprehensive cryptanalysis of the single-degree Legendre PRF operating over $\mathbb{F}_{p^r}$.

First, we analyze the structural properties of polynomial input encoding under a standard passive threat model (sequential additive counter queries). We demonstrate that while the absence of polynomial carry-overs causes an asynchronous "no-carry fracture" that successfully neutralizes classical sliding-window collision attacks, the fracture itself is deterministically periodic. By introducing a novel "Differential Signature" bucketing technique, we show that an attacker can systematically group fractured sequences by their structural shapes. This allows an adversary to bypass the no-carry defense entirely and passively recover the secret key in $\mathcal{O}(U \cdot p^r/M)$ operations, where $U$ is the unicity distance.

Second, we evaluate the PRF under an active Chosen-Query threat model. We demonstrate that an adversary can completely circumvent the additive fracture by evaluating the PRF along a geometric sequence generated by a primitive polynomial. This query structure invokes strict multiplicative homomorphism over $\mathbb{F}_{p^r}^*$, permitting the direct generalization of the state-of-the-art table collision attack to extract the key in purely $\mathcal{O}(p^r/M)$ operations. Finally, we evaluate the cryptographic boundaries of these attacks, formally establishing the absolute necessity of higher-degree key variants ($d \ge 2$) to achieve exponential security against structural reduction in extension fields.
\end{abstract}
\section{Introduction and Mathematical Foundations}
The Legendre Pseudorandom Function (PRF) was originally proposed as a highly efficient primitive designed to minimize multiplicative complexity in Multi-Party Computation (MPC) and symmetric Zero-Knowledge Proof (ZKP) systems \cite{grassi2016}. While the baseline construction operates over prime fields $\mathbb{F}_p$ \cite{damgard1988}, optimizing throughput in specific cryptographic protocols often necessitates shifting operations to extension fields $\mathbb{F}_{p^r}.$ In this section, we formalize the algebraic foundations required to construct and evaluate the Legendre PRF over these extension fields.

\subsection{The Standard Legendre PRF ($d = 1$)}
In its original, single-degree instantiation over a prime field $\mathbb{F}_p$ (where $p$ is an odd prime), the Legendre PRF is evaluated using a secret key $K \in \mathbb{F}_p$ and an input $x \in \mathbb{F}_p.$ The function maps the result of the standard Legendre symbol to a single binary bit:   
\begin{equation*}
    L_K(x) = \left(\frac{x + K}{p}\right)
\end{equation*}

where the mathematical outputs $+1$ (quadratic residue) and $-1$ (quadratic non-residue) are mapped to the binary values $0$ and $1$, respectively. The simplicity of the linear shift $(x + K)$ is the primary source of the function's MPC efficiency, but it also forms the structural basis for sliding-window collision attacks \cite{khovratovich2019,beullens2019}. 

\subsection{Extension Field Construction ($\mathbb{F}_{p^r}$)}
To port this function to an extension field, we transition from scalar integer arithmetic to polynomial arithmetic. The field $\mathbb{F}_{p^r}$ is defined by a base prime $p$ and an extension degree $r$.

The elements of this field are polynomials with a maximum degree of $r - 1$, taking the form: 
\begin{equation*}
    A(x) = a_{r-1}x^{r-1} + \cdots + a_1x + a_0
\end{equation*}
where each coefficient $a_i \in \mathbb{F}_p.$

To ensure the field is closed under multiplication, operations are generally performed modulo an irreducible polynomial $I(x)$ of degree exactly $r$. Therefore, for any two elements $A(x), B(x) \in \mathbb{F}_{p^r},$ field addition and multiplication are defined as:
\begin{equation*}
    A(x) \oplus B(x) \equiv (A(x) + B(x)) \pmod p
\end{equation*}

\begin{equation*}
    A(x)\otimes B(x) \equiv (A(x)\cdot B(x)) \pmod {I(x), p}
\end{equation*}

Crucially, polynomial addition over $\mathbb{F}_p$ operates strictly coefficient-wise with no \enquote{carry-over} to higher degrees, a property that fundamentally alters the behavior of sequential additive queries.

\subsection{The Generalized Legendre Symbol}
To evaluate the PRF, we must generalize the Legendre symbol for $\mathbb{F}_{p^r}.$ An element $A(x) \in \mathbb{F}^*_{p^r}$ is considered a quadratic residue if there exists some element $Y(x)$ such that $Y(x)^2 \equiv A(x) \pmod {I(x), p}$.

By Euler's Criterion generalized for finite fields \cite{lidl1997}, the Legendre symbol of a polynomial $A(x)$ can be deterministically computed via exponentiation:
\begin{equation*}
    L(A(x)) \equiv A(x)^{\frac{p^r - 1}{2}} \pmod{ I(x), p}
\end{equation*}

This operation yields $+1$ if $A(x)$ is a quadratic residue, and $-1$ if it is a non-residue. For the PRF output, we maintain the standard binary mapping:
\begin{equation*}
\mathrm{PRF}(A) =
\begin{cases}
0 & \text{if } L(A(x)) \in \{0,1\} \\
1 & \text{if } L(A(x)) = -1
\end{cases}
\end{equation*}

\section{Implementation Constraints and Statistical Security}
Before evaluating the structural vulnerabilities of the Legendre PRF, we must rigorously define the mapping between sequential integer queries and the extension field $\mathbb{F}_{p^r}$, and formally bound the statistical pseudorandomness of the resulting keystream.

\subsection{Polynomial Input Encoding and Periodicity} \label{Section 2.1}
In standard Counter Mode, the server generates a keystream by evaluating the PRF on a sequentially incrementing integer counter $n \in \mathbb{N}$. To evaluate the function over an extension field, this integer must be deterministically injected into the polynomial ring.

\begin{definition}[Base-$p$ Polynomial Encoding] \label{def:encoding}
Let the extension field $\mathbb{F}_{p^r}$ consist of polynomials of maximum degree $r-1$. We define the encoding map $\phi: \mathbb{Z}_{p^r} \to \mathbb{F}_{p^r}$ such that for an integer $n$ with base-$p$ expansion $n = \sum_{i=0}^{r-1} d_i p^i$, the corresponding field element is:
\begin{equation*}
    \phi(n) = X_n(x) = d_{r - 1}x^{r - 1} + \cdots + d_1x + d_0
\end{equation*}
where $d_i \in \mathbb{F}_p$.
\end{definition}

This encoding enforces a strict physical boundary on the allowable query space, as the map $\phi$ is only bijective over the finite interval $[0, p^r - 1]$.

\begin{proposition}[Periodicity Constraint]
Let $M$ be the total number of sequential queries. If $M \geq p^r$, the input sequence $X_n \pmod{p^r}$ silently wraps, rendering the resulting PRF keystream perfectly periodic with a maximum period of $p^r$.
\end{proposition}

Consequently, for cryptographic viability, the extension field parameters must be selected such that $p^r \gg M$, ensuring the adversary never observes a periodic cycle.

\subsection{Pseudorandomness and the Extension Field Weil Bound}
To ensure the PRF resists statistical distinguishing attacks, the generated bitstream must approximate a uniform distribution. Over prime fields, the distribution of Legendre sequences is bounded by the Weil theorem for multiplicative characters \cite{weil1948}. We extend this bound to the generalized Legendre symbol over $\mathbb{F}_{p^r}$.

\begin{lemma}[Weil Bound for Multiplicative Characters over $\mathbb{F}_{p^r}$] \label{lemma:weil}
Let $\chi$ be the quadratic character of $\mathbb{F}_{p^r}$ (i.e., the generalized Legendre symbol $L$), and let $f(T) \in \mathbb{F}_{p^r}[T]$ be a polynomial of degree $d \geq 1$ that is not a perfect square in $\mathbb{F}_{p^r}[T]$. Then the character sum over $\mathbb{F}_{p^r}$ satisfies:
\begin{equation*}
    \left| \sum_{A \in \mathbb{F}_{p^r}} \chi(f(A)) \right| \leq (d - 1)\sqrt{p^r}
\end{equation*}
where the sum excludes the $d$ roots of $f$ (at which $\chi$ vanishes), and the bound holds for squarefree $f$ of degree $d \geq 2$. For the PRF's linear input $f(A) = A + K$ (degree $d = 1$), the map $A \mapsto A + K$ is a bijection on $\mathbb{F}_{p^r}$, so the character sum $\sum_{A} L(A + K)$ equals $\sum_{A \in \mathbb{F}_{p^r}} L(A) = 0$ exactly, by the orthogonality of characters.
\end{lemma}

Using Lemma \ref{lemma:weil}, we can formally bound the expected frequency of any sliding-window pattern within the keystream, proving its resistance to statistical biases.

\begin{theorem}[Statistical Indistinguishability]
Let $S$ be a PRF keystream generated over $\mathbb{F}_{p^r}$. If an adversary analyzes sliding-window patterns of length $l$, the absolute frequency of any specific $l$-bit binary sequence deviates from the uniform expectation by at most $\mathcal{O}(\sqrt{p^r})$. Specifically, the expected occurrence $N_l$ of the pattern in the query space satisfies:
\begin{equation*}
    N_l = \frac{p^r}{2^l} + \mathcal{O}(\sqrt{p^r})
\end{equation*}
The error term $\mathcal{O}(\sqrt{p^r})$ follows from applying the Weil bound to degree-$d$ polynomial combinations ($d \geq 2$) that arise when counting joint occurrences of $l$-bit patterns; for $l \geq 2$, the relevant character sums involve polynomials of degree $l \geq 2$, to which the bound $(d-1)\sqrt{p^r} \leq (l-1)\sqrt{p^r} = \mathcal{O}(\sqrt{p^r})$ applies correctly.
\end{theorem}

As the field size $p^r \to \infty$, the fractional error term bounded by the Weil theorem becomes negligible, and the distribution converges perfectly to the uniform distribution. Therefore, any successful cryptanalysis of the Legendre PRF over $\mathbb{F}_{p^r}$ cannot rely on statistical bias, but must instead exploit the underlying algebraic structure of the evaluation.

\section{Structural Analysis and Passive Cryptanalysis} \label{Section 3}

The security of the Legendre PRF against classical collision attacks \cite{khovratovich2019} relies strictly on the continuity of the input sequence differential. In this section, we formally evaluate the structural behavior of polynomial base-$p$ encoding and prove that the sequence can be compromised passively despite the failure of constant arithmetic progressions.

\subsection{The Additive Differential and Sequence Fracture}
In a passive threat model, the server generates a keystream $S = \{ L_K(X_n) \}_{n=1}^M$ by evaluating the PRF on a sequentially incrementing counter $n$, mapped to polynomials $X_n \in \mathbb{F}_{p^r}$ as defined in Section \ref{Section 2.1}. 

Classical table collision attacks over $\mathbb{F}_p$ necessitate a constant differential $\Delta_n = X_{n+1} - X_n = 1$. We demonstrate that the polynomial mapping over extension fields inherently fractures this progression.

\begin{lemma}[Asynchronous Polynomial Carry] \label{lemma:carry}
Let $X_n \in \mathbb{F}_{p^r}$ be the polynomial encoding of an integer $n$. The differential $\Delta_n = X_{n+1} \oplus X_n \pmod p$ is not constant, but rather takes values from the set
\begin{equation*}
    \left\{ \sum_{i=0}^{k} x^i \;\middle|\; 0 \leq k < r \right\},
\end{equation*}
determined strictly by the $p$-adic valuation of $n+1$.
\end{lemma}

\begin{proof}
Let $n = \sum_{i=0}^{r-1} d_i p^i$ mapping to $X_n = \sum_{i=0}^{r-1} d_i x^i$. The standard integer addition $n + 1$ induces a carry cascading up to the $k$-th digit if and only if $d_i = p - 1$ for all $0 \leq i < k$ (and $d_k \neq p-1$). In that case, incrementing $n$ sets $d_i \to 0$ for $0 \leq i < k$ and $d_k \to d_k + 1$.

Because polynomial addition in $\mathbb{F}_{p^r}$ operates coefficient-wise modulo $p$ without algebraic carry-over, each reset $d_i = p - 1 \to 0$ contributes a coefficient difference of $0 - (p-1) \equiv 1 \pmod{p}$, and the increment $d_k \to d_k + 1$ contributes a difference of $+1$ at degree $k$. Therefore all coefficients $x^0, x^1, \ldots, x^k$ increase by $1$, yielding
\begin{equation*}
    \Delta_n = x^k + x^{k-1} + \cdots + x + 1 = \sum_{i=0}^{k} x^i.
\end{equation*}
When $k = 0$ (no carry, i.e.\ $d_0 \neq p-1$), only the constant coefficient increments and $\Delta_n = 1$, consistent with the formula for $k = 0$.
\end{proof}

Lemma \ref{lemma:carry} proves that the differential $\Delta_n$ takes at most $r$ distinct values, one for each possible carry depth $k \in \{0, \ldots, r-1\}$. Consequently, the input sequence fractures into disjointed arithmetic sub-blocks of length $p$, rendering the contiguous sliding-window reference required by Beullens et al.\ \cite{beullens2019} mathematically impossible.

\subsection{Passive Key Recovery via Differential Signatures}
While the constant progression is fractured, the sequence of differentials $(\Delta_n)_{n \in \mathbb{N}}$ is deterministically periodic. We exploit this determinism to mount a passive recovery attack.

\begin{definition}[Differential Signature] \label{def:signature}
For a target sequence window $W_n$ of unicity length $U$ starting at index $n$, the differential signature $\vec{\delta}_n$ is the sequence of cumulative polynomial differences relative to $X_n$:
\begin{equation*}
    \vec{\delta}_n = (0, X_{n+1} \ominus X_n, X_{n+2} \ominus X_n, \dots, X_{n+U-1} \ominus X_n)
\end{equation*}
\end{definition}

\begin{theorem}[Passive Key Recovery] \label{thm:passive}
Let $M$ be the number of observed keystream bits, and $U$ be the unicity distance such that $2^U \gg p^r$. An adversary can recover the secret key $K$ passively in expected $\mathcal{O}(U \cdot p^r/M)$ operations.
\end{theorem}

\begin{proof}
By Definition \ref{def:signature}, the target window $W_n$ evaluates the PRF as:
\begin{equation*}
    W_n = \left( L(X_n + K + \delta_0), \dots, L(X_n + K + \delta_{U-1}) \right)
\end{equation*}
Letting $Z = X_n + K$, the sequence reduces to $L(Z + \delta_i)$ for $0 \leq i < U$, which depends exclusively on the universal shape $\vec{\delta}_n$ and the scalar polynomial $Z$.

The number of unique differential signatures for a window of length $U$ is strictly bounded by $p^{\lceil \log_p U \rceil}$. The adversary extracts all $M - U + 1$ overlapping windows from the keystream and partitions them into equivalence classes based on their signature $\vec{\delta}$. 

The adversary selects the maximally populated equivalence class, expected to contain $M' \approx M/U$ sequence targets, and constructs a hash table $H$. The adversary then randomly samples $Z \in \mathbb{F}_{p^r}$, evaluates the reference window using the associated signature $\vec{\delta}_{max}$, and queries $H$. A collision occurs with probability $M'/p^r$, requiring $p^r/M'$ iterations. 

Upon collision, the table yields the target index $n$, and the key is recovered algebraically via $K = Z \ominus X_n \pmod p$. The total time complexity is bounded by $\mathcal{O}(p^r / (M/U)) = \mathcal{O}(U \cdot p^r/M)$.
\end{proof}

\section{Active Cryptanalysis via Geometric Sequences}
While Theorem \ref{thm:passive} establishes a passive bypass of the polynomial carry, the complexity suffers a linear penalty proportional to $U$. In this section, we formalize an active chosen-query attack that restores strict algebraic homomorphism, eliminating the fracture entirely.

\subsection{Multiplicative Homomorphism over $\mathbb{F}_{p^r}^*$}
To completely circumvent the additive limitations of $\mathbb{F}_p[x]$, an active adversary with oracle access to the server abandons polynomial addition. Instead, the adversary selects a primitive generator $g(x)$ of the multiplicative group $\mathbb{F}_{p^r}^*$ and queries the server with a geometric sequence $X_i = g^i \pmod{I(x), p}$ for $1 \leq i \leq M$.

\begin{lemma}[Factoring Identity of the Legendre PRF] \label{lemma:factor}
For a geometric query sequence $g^i$, the Legendre PRF evaluation factors into a secret-dependent constant and a universal, keyless reference sequence shifted by a discrete logarithm $j$.
\end{lemma}

\begin{proof}
The Legendre symbol generalized over $\mathbb{F}_{p^r}$ is completely multiplicative. Thus, the evaluation $L(g^i + K)$ can be factored as:
\begin{equation*}
    L(g^i + K) \equiv L\left( K \otimes (K^{-1}g^i + 1) \right) \equiv L(K) \cdot L(K^{-1}g^i + 1) \pmod{I(x), p}
\end{equation*}
Since $g(x)$ generates $\mathbb{F}_{p^r}^*$, there exists a unique integer $j$ such that $g^j \equiv K^{-1} \pmod{I(x), p}$. Substituting $g^j$ yields:
\begin{equation*}
    L(g^i + K) = L(K) \cdot L(g^{i+j} + 1)
\end{equation*}
The sequence is thereby algebraically isolated into an unknown shift $j$ and a sign coefficient $L(K) \in \{-1, +1\}$, operating over the universal reference $L(g^m + 1)$.
\end{proof}

\subsection{Geometric Table Collision Attack}
Leveraging Lemma \ref{lemma:factor}, we generalize the state-of-the-art table collision attack \cite{beullens2019} to operate over the multiplicative group of the extension field.

\begin{theorem}[Active Key Recovery] \label{thm:active}
Given $M$ chosen geometric queries, an adversary can extract the secret polynomial $K \in \mathbb{F}_{p^r}$ in purely $\mathcal{O}(p^r / M)$ operations.
\end{theorem}

\begin{proof}
The adversary intercepts the geometric keystream and populates a hash table $T$ mapping $U$-bit contiguous windows to their starting index $i$. To mitigate spurious collisions, $U$ is chosen such that $2^U \gg p^r$.

The adversary defines the universal reference function $L_0(m) = L(g^m + 1)$. By randomly sampling exponents $m \in \{1, \dots, p^r-1\}$, the adversary generates reference windows of length $U$. By Lemma \ref{lemma:factor}, the target sequence is a strict shift of the reference sequence, potentially inverted if $L(K) = -1$. Therefore, the adversary queries $T$ for both the generated reference window and its bitwise complement.

Because $T$ contains $M - U \approx M$ targets, a collision is guaranteed in expected $p^r/M$ guesses. Let the collision match target index $i$ and reference index $m$. By the factoring identity, the exponent relation is:
\begin{equation*}
    m \equiv i + j \pmod{p^r - 1} \implies j \equiv m - i \pmod{p^r - 1}
\end{equation*}
The secret polynomial $K$ is then extracted via the field inversion:
\begin{equation*}
    K \equiv (g^j)^{-1} \pmod{I(x), p}
\end{equation*}
This process requires exactly $\mathcal{O}(p^r/M)$ Legendre symbol evaluations, effectively reducing the active security of the single-degree PRF in $\mathbb{F}_{p^r}$ to its theoretical lower bound.
\end{proof}

\section{Cryptographic Boundaries and Future Work}
The cryptanalysis presented in Sections 3 and 4 demonstrates that the single-degree ($d=1$) Legendre PRF is fundamentally insecure over extension fields under both passive and active threat models. However, the original PRF framework \cite{khovratovich2019} allows for the use of higher-degree polynomials to mitigate algebraic weaknesses. As part of our future work, we intend to formally evaluate the security boundaries of these generalized constructions.

\subsection{The Higher-Degree Variant ($d \geq 2$)}
To defend against structural reduction, the secret key can be expanded from a single element into a set of $d$ independent polynomials: $K = (K_{d-1}, K_{d-2}, \dots, K_0)$. The PRF evaluates a monic polynomial of degree $d$:
\begin{equation*}
L_K(x) = L(x^d + K_{d-1}x^{d-1} + \cdots + K_1x + K_0)
\end{equation*}

For example, the degree-2 (quadratic) variant evaluates:
\begin{equation*}
L_K(x) = L(x^2 + K_1x + K_0)
\end{equation*}

When an adversary attempts to execute the active geometric attack (Section 4) by querying $X_i = g^i$, the evaluation becomes:

\begin{equation*}
L( (g^i)^2 + K_1(g^i) + K_0 )
\end{equation*}
Unlike the single-degree variant, it is mathematically impossible to factor this expression into a secret-dependent constant and a universal reference sequence. The presence of the $K_1x$ term inherently binds the sequence to the specific, unknown combination of $(K_1, K_0)$. Because the variables cannot be cleanly isolated, the $\mathcal{O}(p^r/M)$ table collision attack is effectively neutralized.

\subsection{Future Cryptanalysis and Security Bounds}

While the geometric sequence attack fails against $d \geq 2$, the algebraic structure of the extension field may still permit advanced reduction techniques. Future research will explore the application of Tschirnhaus transformations and ``completing the square'' methodologies to eliminate intermediate key coefficients.

Preliminary analysis suggests that by executing Meet-in-the-Middle (MITM) attacks over these algebraically reduced forms, an adversary can decrease the brute-force complexity from $\mathcal{O}(p^{dr})$ down to approximately $\mathcal{O}(p^{\lfloor d/2 \rfloor r})$. Consequently, while the degree-2 variant may still be vulnerable to optimized searches in $\mathcal{O}(p^r)$ time, deploying the Legendre PRF with degrees $d \geq 3$ or $4$ appears necessary to establish a robust, exponential security margin against structural reduction in $\mathbb{F}_{p^r}$.

\section*{Availability of Source Code}
The Python implementations for all cryptanalytic attacks described in this paper, including the randomized geometric table collision and the passive shape-bucketing optimization, are open-source and publicly available on GitHub at: 
\begin{center}
    \url{https://github.com/DA1729/lprf\_ext\_fields\_analysis}.
\end{center}

\section{Conclusion}
In this paper, we presented a comprehensive cryptanalysis of the single-degree Legendre Pseudorandom Function (PRF) operating over extension fields $\mathbb{F}_{p^r}$. We established that while the transition from prime fields to extension fields disrupts the continuity of standard additive counter queries—resulting in an asynchronous \enquote{no-carry fracture} that effectively thwarts classical sliding-window collision attacks—this structural boundary is fundamentally insufficient to guarantee cryptographic security.

Under a standard passive threat model, we demonstrated that the deterministic periodicity of polynomial input encoding allows an adversary to cluster fractured sequence windows by their differential signatures. By leveraging this novel shape-bucketing optimization, the adversary can bypass the additive fracture entirely and passively recover the secret key in $\mathcal{O}(U \cdot p^r/M)$ operations.

Furthermore, by elevating the threat model to a Chosen-Query Attack, we proved that an active adversary can completely circumvent the additive group limitations. By querying the PRF along a geometric sequence generated by a primitive polynomial, the adversary invokes strict multiplicative homomorphism over $\mathbb{F}_{p^r}^*$. This mathematical reduction factors the secret key into a pure exponent shift, permitting the direct generalization of the state-of-the-art table collision attack \cite{beullens2019} and extracting the key in precisely $\mathcal{O}(p^r/M)$ operations.

Ultimately, our findings formally break the $d=1$ Legendre PRF over extension fields under both passive and active threat models. To maintain an exponential security margin against structural and algebraic reduction, future deployments of the Legendre PRF in $\mathbb{F}_{p^r}$ for Multi-Party Computation and Zero-Knowledge applications must inherently rely on higher-degree key variants ($d \geq 2$) to mathematically bind the polynomial coefficients and resist generalized collision cryptanalysis.

\end{document}